\documentclass[11pt]{article}
\usepackage[utf8]{inputenc}

\usepackage[backend=biber]{biblatex}
\usepackage{
  fullpage, microtype, authblk,
  booktabs, diagbox, algorithm,
  amsmath, amssymb, amsthm, bm,
  url, tikz, subcaption, relsize,
}
\usepackage[font=small,labelfont=bf]{caption}
\usepackage[inline]{enumitem}
\usepackage[normalem]{ulem}
\usepackage[noend]{algpseudocode}
\usepackage{xspace}

\usetikzlibrary{arrows.meta, calc, shapes, positioning}
\MakeRobust{\Call}

\newcommand{\mund}{\mathunderscore}
\newcommand{\R}{\ensuremath{\mathcal{R}}\xspace}
\newcommand{\G}{\ensuremath{\mathcal{G}}\xspace}
\newcommand{\I}{\ensuremath{\mathcal{I}}\xspace}
\newcommand{\T}{\ensuremath{\mathcal{T}}\xspace}
\newcommand{\J}{\ensuremath{\mathcal{J}}\xspace}
\newcommand{\journeys}{\J\xspace}
\newcommand{\TC}{{\footnotesize TC}\xspace}
\newcommand{\TCs}{{\footnotesize TC}s\xspace}
\newcommand{\TTC}{{\footnotesize TTC}\xspace}
\newcommand{\TTCs}{{\footnotesize TTC}s\xspace}
\newcommand{\Rtuple}{{\footnotesize R}-tuple\xspace}
\newcommand{\Rtuples}{{\footnotesize R}-tuples\xspace}
\newcommand{\alga}{\texttt{{add\mund{}contact(u, v, t)}}\xspace}
\newcommand{\algb}{\texttt{can\mund{}reach(u, v, $t_1$, $t_2$)}\xspace}
\newcommand{\algc}{\texttt{is\mund{}connected($t_1$, $t_2$)}\xspace}
\newcommand{\algd}{\texttt{reconstruct\mund{}journey(u, v, $t_1$, $t_2$)}\xspace}
\newcommand{\alganame}{\texttt{{add\mund{}contact}}\xspace}
\newcommand{\algbname}{\texttt{can\mund{}reach}\xspace}
\newcommand{\algcname}{\texttt{is\mund{}connected}\xspace}
\newcommand{\algdname}{\texttt{reconstruct\mund{}journey}\xspace}

\newtheorem{definition}{Definition}
\newtheorem{observation}{Observation}
\newtheorem{theorem}{Theorem}

\newtheorem{lemma}[theorem]{Lemma}

\tikzset{
  vertex/.style={
    draw,
    fill=white,
    circle,
    semithick,
    inner sep=0pt,
    minimum size=7pt,
    font=\footnotesize,
  },
  arc/.style={
    -{Latex[length=5pt]},
    semithick,
    font=\scriptsize,
  },
}

\title{
  A Dynamic Data Structure for Temporal Reachability\\ with Unsorted Contact Insertions
}

\author[1,2]{Luiz F. Afra Brito}
\author[1]{Marcelo Albertini}
\author[2]{Arnaud Casteigts}
\author[1]{Bruno A. N. Traven\c{c}olo}

\affil[1]{Federal University of Uberlândia, Brazil}
\affil[2]{University of Bordeaux, France}

\begin{document}

\maketitle
\begin{abstract}
  Temporal graphs represent interactions between entities over the time. These interactions may be direct (a contact between two nodes at some time instant), or indirect, through sequences of contacts called temporal paths (journeys).
    Deciding whether an entity can reach another through a journey is useful for various applications in communication networks and epidemiology, among other fields.
    In this paper, we present a data structure which maintains temporal reachability information under the addition of new contacts (\textit{i.e.,} triplets $(u,v,t)$ indicating that node $u$ and node $v$ interacted at time $t$). In contrast to previous works, the contacts can be inserted in arbitrary order---in particular, non-chronologically---which corresponds to systems where the information is collected a posteriori (e.g. when trying to reconstruct contamination chains among people). 
    The main component of our data structure is a generalization of transitive closure called \emph{timed transitive closure} (\TTC), which allows us to maintain reachability information relative to all nested time intervals, without storing all these intervals, nor the journeys themselves. \TTCs are of independent interest and we study a number of their general properties.
    Let $n$ be the number of nodes and $\tau$ be the number of timestamps in the \emph{lifetime} of the temporal graph. Our data structure answers reachability queries regarding the existence of a journey from a given node to another within given time interval in time
    $O(\log\tau)$; it has an amortized insertion time of $O(n^2\log\tau)$; and it can reconstruct a valid journey that witnesses reachability in time $O(k\log\tau)$, where $k<n$ is the maximum number of edges of this journey.
    Finally, the space complexity of our reachability data structure is $O(n^2\tau)$, which remains within the worst-case size of the temporal graph itself.

\end{abstract}

\section{Introduction}

Temporal graphs represent interactions between entities over the time.
These interactions often appear in the form of contacts at specific timestamps.
Moreover, entities can also interact indirectly with each others by chaining several contacts over time.
For example, in a communication network, devices that are physically connected can send new messages or propagate received ones; thus, by first sending a new message and, then, repeatedly propagating messages over time, remote entities can communicate indirectly.
Time-respecting paths in temporal graphs are known as temporal paths, or simply \emph{journeys}, and when a journey exists from one node to another node, we say that the first can \emph{reach} the second.

In a computational environment, it is often useful to check whether entities can reach each other.
Investigations on temporal reachability have been used
for characterizing mobile and social networks~\cite{tang2010characterising};
for validating protocols and better understanding communication networks~\cite{cacciari1996atemporal, whitbeck2012temporal};
for checking the existence of trajectories and improving flow in transportation networks~\cite{wu2017mining, williams2016spatio, bedogni2018temporal};
for assessing future states of ecological networks~\cite{martensen2017spatio};
and for making plans for agents using automation networks~\cite{bryce2007atutorial}.
Beyond the sole reachability, some applications require the ability to reconstruct a concrete journey if one exists.
For example, journey reconstruction has been used
for finding and visualizing detailed trajectories in transportation networks~\cite{wu2017mining, betsy2007spatio, zeng2014visualizing, hasan2011making};
for visualizing system dynamics~\cite{hurter2014bundled};
and for matching temporal patterns in temporal graph databases~\cite{vera2016querying,LVM18}.

In standard graphs, the problem of maintaining reachability information under various modifications of the graph is known as \emph{dynamic connectivity} and has been extensively studied~\cite{agrawal1989efficient, haixun2006dual, cohen2003reachability,zhu2014reachability, seufert2013ferrari,wei2018reachability}. Here, the adjective \emph{dynamic} does not refer to the temporal nature of the network, it refers to the fact that the computed information is to be updated after the input graph is changed. This maintenance is performed by means of a \emph{dynamic data structure}, which stores intermediate information to speed up the query time (and its own update time after a change). Three types of dynamic data structures are classically considered, depending on the type of change allowed, namely \emph{incremental} (insertion only), \emph{decremental} (deletion only), and \emph{fully-dynamic} (both). Typically, the elements to be inserted and removed in the classical version are the edges of the graph.

In the case of temporal graphs, the elements to be inserted are not edges but contacts, which are edges together with a timestamp, which indicates that the two corresponding nodes interacted at this particular time. An important aspect of data structures which manipulate contacts is whether the order of insertion respects the order of the interactions themselves, i.e., the insertions are \emph{chronological}. Algorithms for updating reachability information with the assumption that the input is chronological have been proposed~\cite{barjon2014testing,whitbeck2012temporal}. However, this assumption does not capture important use cases where the contacts are collected in an unpredictable order and the reachability information is updated afterwards.
For instance, during scenarios of epidemics, outdated information containing the interaction details among infected and non-infected individuals are reported in arbitrary order.
Then, this information is periodically queried in order to better understand the dissemination process, and then take appropriate measures for reducing reachability and identifying sources of contamination~\cite{xiao2018reconstructing, xiao2018robust,enright2018deleting,rozenshtein2016reconstructing}.


Motivated by such scenarios, we investigate the problem of maintaining an incremental data structure for temporal reachability, where the insertions of contacts are made in arbitrary order.
A naive approach is to store and update the temporal graph itself (e.g., as a set of contact), then run standard journey computation algorithms like~\cite{xuan2003computing} when a query is made. However, the goal of a data structure is to reduce the computational cost of the queries based on pre-computing intermediate information. In fact, data structures typically offer a tradeoff between query time, update time, and space. To the best of our knowledge, the only existing work supporting non-chronological contact insertions and exploiting intermediate representations for speeding up reachability queries in temporal graphs is~\cite{wu2016reachability}. The solution in~\cite{wu2016reachability} relies on maintaining a directed acyclic graph (DAG) in which every original vertex is possibly copied up to $\tau$ times (where $\tau$ is the number of timestamps) and a journey exists from $u$ to $v$ in the interval $[t, t']$ if and only if and vertex $u_t$ can reach vertex $v_{t'}$ in the DAG. Some paths preprocessing is additionally considered that results in an average speed up for reachability queries. However, the worst-case query time corresponds to a standard path search (e.g. depth first search) in the DAG, which takes $\Theta(n^2\tau)$ time in the case of dense temporal graphs (whose number of contacts is of the same order). The space complexity (size of the DAG) also corresponds essentially to the number of contacts, thus $\Theta(n^2\tau)$ in the worst case. Finally, the update time upon insertion is quite efficient, because the DAG representation allows its effect to remain local. If one ignores the cost of paths preprocessing in~\cite{wu2016reachability} (as we focus on worst-case analysis), it only takes $O(1)$ time to update the DAG if the corresponding nodes are already known, and up to $\Theta(\tau)$ otherwise, due to the creation of (up to) $\tau$ copies of the new nodes.

\subsection{A data structure for unsorted contact insertions}
\label{sec:primitives}
In this paper, we consider the problem of maintaining reachability information through a data structure\footnote{We have a simple implementation available at \url{https://github.com/albertiniufu/dynamictemporalgraph/}} that supports the following four operations, where by convention, $\mathcal{G}$ is a temporal graph, $u$ and $v$ are vertices of $\mathcal{G}$, and $t, t_1,$ and $t_2$ are timestamps:
\begin{itemize}
    \item \alga: Update information based on a contact from $u$ to $v$ at time $t$
    \item \algb: Return true if $u$ can reach $v$ within the interval $[t_1, t_2]$
    \item \algc: Return true if $\mathcal{G}$ restricted to the interval $[t_1,t_2]$ is temporally connected, \textit{i.e.} all vertices can reach each other within the interval $[t_1, t_2]$
    \item \algd: Return a journey (if one exists) from $u$ to $v$ occurring within the interval $[t_1, t_2]$
\end{itemize}

For generality, we consider directed contacts (timed arcs). Furthermore, if $t_1$ and $t_2$ are omitted in the above operations, then the entire lifetime of $\mathcal{G}$ is considered.
The challenge in realizing these operations is to answer queries as fast as possible, while keeping space consumption and update time at reasonable levels. The worst-case complexities are as follows: the query operations, \algb{} and \algc{}, run, respectively, in $O(\log\tau)$ and $O(n^2\log\tau)$ time; the
update operation, \alga{}, runs in $O(n^2\log\tau)$ amortized time;
and the retrieval operation, \algd{}, runs in $O(k\log\tau)$ time, where $k < n$ is the length of the resulting journey. The worst-case space complexity remains within the worst-case size of the temporal graph itself, namely $O(n^2\tau)$. Overall, the space complexity is comparable to that of~\cite{wu2016reachability}, while the query time is much faster and the update time is slower.

The core of our data structure is a component called the \emph{timed transitive closure} (\TTC), which generalizes the classical notion of a transitive closure (\TC). Classical TCs capture reachability information among vertices over the entire lifetime of the network. They are classically encoded as a static directed graph where the existence of an arc from $u$ to $v$ implies that there is a journey $u$ to $v$ in the temporal graph. If one is not interested in querying reachability for specific subintervals, and if the contacts are inserted in chronological order, then \TCs are actually sufficient for maintaining temporal reachability information (see e.g.~\cite{barjon2014testing}). A generalization of \TC has also been considered in~\cite{whitbeck2012temporal}, which allows queries to be parametrized by a maximum journey duration, however basic journey information, such as departure and arrival times, are not known and the computation of the structure requires the information to be processed at once and chronologically (i.e. subsequent updates are not supported).

In the unsorted (i.e., non-chronological) case, \TCs do not provide enough information to decide whether a new contact (possibly occurring at any point in history) can be composed with known journeys. To address this need, we introduce a generalization of \TCs called \emph{timed transitive closure} (\TTCs), which store information regarding the availability of journeys for a well-chosen set of time intervals, without storing the journeys themselves. We study the general properties of \TTCs and we prove, in particular, that one can restrict the number of intervals considered to $O(\tau)$ for any pair of nodes (as opposed to $O(\tau^2)$), with immediate consequences on the space complexity of a data structure based on \TTCs. This information is then exploited by our data structure algorithms.

\subsection{Organization of the document}
This paper is organized as follows.
In Section~\ref{sec:definitions}, we present basic definitions.
In Section~\ref{sec:timed-transitive-closure}, we introduce timed transitive closures, study their basic properties, and provide a number of low-level primitives for manipulating them.
In Section~\ref{sec:data-structure}, we describe the algorithms that perform each operation of our data structure based on \TTCs, together with their running time complexities.
Finally, Section~\ref{sec:conclusions} concludes with some remarks and open questions.

\section{Definitions}\label{sec:definitions}

Following the formalism in~\cite{casteigts2012time}, a temporal graph can be generally represented by a tuple $\mathcal{G} = (V, E, \mathcal{T}, \rho, \zeta)$, where $V$ is a set of vertices,  $E \subseteq{} V \times V$ is a set of edges,
 $\mathcal{T}$ is the time interval over which the network exists (\emph{lifetime}),
 $\rho: E \times \mathcal{T} \to\{0, 1\}$ is a \emph{presence function} that expresses whether a given edge is present at a given time instant,
and $\zeta: E \times\mathcal{T} \mapsto \mathbb{T}$ is a latency function that expresses the duration of an interaction for a given edge at a given time, where $\mathbb{T}$ is the time domain (typically $\mathbb{R}$ or $\mathbb{N}$).
In this paper, we consider a setting where $E$ is a set of arcs (directed edges), $\mathbb{T}$ is equal to $\mathbb{N}$  (time is discrete) and  $\mathcal{T} = [1, \tau] \subseteq \mathbb{T}$ (the lifetime contains $\tau$ timestamps). The latency function is constant, namely $\zeta = \delta$, where $\delta$ is any fixed positive integer (typically $0$ or $1$).
We call $(u, v, t)$ a \emph{contact} in ${\mathcal{G}}$ if $\rho((u, v), t) = 1$.
We use a short-hand notation $\mathcal{G}_{[t_1, t_k]}$ when restricting the lifetime of ${\mathcal G}$ to a subinterval $[t_1, t_k]\subseteq \mathcal{T}$, and call $\mathcal{G}_{[t_1, t_k]}$ a temporal subgraph of $\mathcal{G}$. Finally, the static graph $G=(V,E)$ is called the \emph{underlying graph} of $\G$.

Reachability in temporal graphs is defined in a time-respecting way, by requiring that a path travels along increasing times (resp. non-decreasing times) if $\delta \ge 1$ (resp. if $\delta=0$). These paths are called temporal paths or journeys, interchangeably.

\begin{definition}[Journey]\label{def:journey}
  A journey from $u$ to $v$ in $\G$ is a sequence of contacts $\mathcal{J} = \langle c_1, c_2, \ldots, c_k \rangle$, whose sequence of underlying arcs form a valid $(u,v)$-path in the underlying graph $G$ and for each contact $c_i = (u_i, v_i, t_i)$, it holds that $\rho((u_i, v_i), t_i) = 1$ and $t_{i+1} \ge t_i + \delta$ for each $i \in [1, k-1]$.
  Additionally, we say that $departure (\mathcal{J}) = t_1$, $arrival (\mathcal{J}) = t_{k} + \delta$ and $duration (\mathcal{J}) = arrival (\mathcal{J})- departure (\mathcal{J})$.
\end{definition}

A journey is \emph{trivial} if it consists of a single contact.


\begin{definition}[Reachability]%
    \label{def:reachability}
    A vertex $u$ can \emph{reach} a vertex $v$ within time interval $[t_1, t_2]$ iff there exists a journey $\J$ from $u$ to $v$ in $\mathcal{G}_{[t_1, t_2]}$ (i.e., such that $departure(\J) \ge t_1$ and $arrival(\J) \le t_2$).
\end{definition}

The (standard) \emph{transitive closure} (\TC) of a temporal graph $\mathcal{G}$ is a static directed graph $\mathcal{G}^*=(V, E^*)$ such that $(u, v) \in E^*$ if and only if $u$ can reach $v$ in $\mathcal{G}$.
This notion is illustrated in Figure~\ref{fig:transitive-closure}.
As already explained, if the contacts are discovered in a chronological order, then $\mathcal{G}^*$ can be updated incrementally to compute the entire reachability information of $\G$~\cite{barjon2014testing}. In the unsorted case, however, this notion is not sufficient because it does not allows one to decide if a new contact can be composed with previously-known journeys, which motivates the definition of more powerful objects.

\begin{figure}[h]
  \centering
  \begin{tikzpicture}[xscale=3, yscale=2.5]
    \node[vertex] at (162:8mm) (a) {}
    node[left = -2pt of a] {$a$};
    \node[vertex] at (90:7mm)  (b) {}
    node[above= -2pt of b] {$b$};
    \node[vertex] at (270:2mm) (c) {}
    node[below= -2pt of c] {$c$};
    \node[vertex] at (378:8mm) (d) {}
    node[right= -2pt of d] {$d$};
    \tikzstyle{every node}=[inner sep=1pt]
    \draw (a) edge[arc] node[above left] {$2$} (b);
    \draw (b) edge[arc, bend left=15] node[above right] {$4$} (d);
    \draw (b) edge[arc, bend right=15] node[below left] {$1$} (d);
    \draw (a) edge[arc, bend left=15] node[above right] {$4$} (c);
    \draw (c) edge[arc, bend left=15] node[below left] {$4$} (a);
    \draw (c) edge[arc] node[below right] {$5$} (d);
  \end{tikzpicture}
  ~~~
  \begin{tikzpicture}[xscale=3, yscale=2.5]
    \node[vertex] at (162:8mm) (a) {}
    node[left = -2pt of a] {$a$};
    \node[vertex] at (90:7mm)  (b) {}
    node[above= -2pt of b] {$b$};
    \node[vertex] at (270:2mm) (c) {}
    node[below= -2pt of c] {$c$};
    \node[vertex] at (378:8mm) (d) {}
    node[right= -2pt of d] {$d$};
    \draw (a) edge[arc] (b);
    \draw (b) edge[arc] (d);

    \draw (a) edge[arc, bend left=15] (c);
    \draw (c) edge[arc, bend left=15] (a);
    \draw (c) edge[arc] (d);

    \draw (a) edge[arc, bend left=0] (d);
  \end{tikzpicture}
  \caption{\label{fig:temporal-graph}
    (Left) A temporal graph $\G$ on four vertices $V = \{ a, b, c, d \}$, where the presence times of the arcs are depicted by labels. Whether $\delta = 0$ or $1$, this graph has only two non-trivial journeys, namely $\mathcal{J}_1 = \langle(a, b, 2), (b, d, 4)\rangle$ and $\mathcal{J}_2 = \langle(a, c, 4), (c, d, 5)\rangle$.
    (Right) Transitive closure $\mathcal{G}^*$. Note that $\mathcal{J}_1$ and $\mathcal{J}_2$ are represented by the same arc in $\G^*$ (the two contacts from $b$ to $d$ as well).}\label{fig:transitive-closure}
\end{figure}
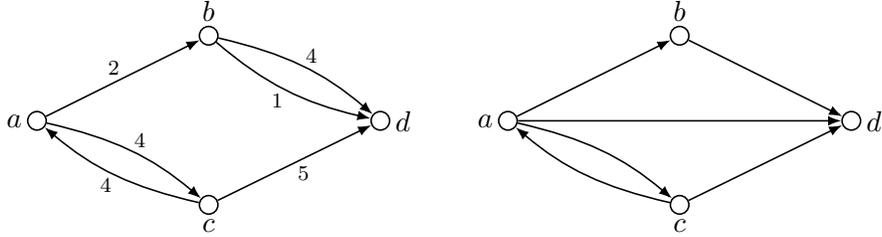

\section{Reachability tuples and Timed Transitive Closure}\label{sec:timed-transitive-closure}

In this section, we describe an extension of the concept of transitive closure called \emph{timed transitive closure} (\TTC).
The purpose of \TTCs is to encode reachability information among the vertices, parametrized by time intervals, so that one can subsequently decide if a new contact occurring anywhere in history can be composed with existing journeys.
The main components of \TTCs are called \emph{reachability tuples} (\Rtuples).
We introduce a number of operators on \Rtuples{}, such as inclusion and concatenation, and describe their role in the construction and maintenance of a \TTC{}.

\subsection{Reachability tuples (\Rtuples)}\label{ssec:R-tuples}

Just as the number of paths in a static graph, the number of journeys in a temporal graph could be too large to be stored explicitly (typically, factorial in $n$). To avoid this problem, \Rtuples capture the fact that a node can reach another within a certain time interval without storing the corresponding journeys. Thus, a single \Rtuple may capture the reachability information corresponding to many possible journeys. We distinguish between two versions of \Rtuples, namely (\emph{existential}) \Rtuples and \emph{constructive} \Rtuples, the latter adding information for reconstructing a journey that witnesses reachability.

\subsubsection{Existential \Rtuples}
The following definitions are given in the context of a temporal graph $\mathcal{G}$ whose vertex set is $V$, lifetime is $\mathcal{T}=[1,\tau]$, and latency is $\delta$.

\begin{definition}[\Rtuple] An (existential) \Rtuple{} is a quadruplet $r=(u, v, t^-, t^+)$, where $u$ and $v$ are vertices in $\mathcal{G}$, and $t^-$ and $t^+$ are timestamps in $\mathcal{T}$. It encodes the fact that node $u$ can reach node $v$ through a journey $\mathcal{J}$ such that $departure (\mathcal{J}) = t^-$ and $arrival (\mathcal{J}) = t^+$.
  If several such journeys exist, then they are all captured by the same \Rtuple.
\end{definition}

The set of journeys captured by an \Rtuple $r$ is denoted by $\journeys(r)$, and we say that $r$ \emph{represents} these journeys.
An \Rtuple{} is \emph{trivial} when it represents a trivial journey (\textit{i.e.}, a single contact). Trivial \Rtuples thus have the form $(u, v, t, t+\delta)$ for some $t$.
The following relations and operations are quite natural to define.

\begin{definition}[Precedence $\prec$]\label{def:tuple-precedence}
  An interval $I_1=[t^-_1, t^+_1]$ precedes an interval $I_2=[t^-_2, t^+_2]$, noted $I_1 \prec I_2$, if $t^+_1 \leq{} t^-_2$.
  Given two \Rtuples{} $r_1 = (u_1, v_1, t^-_1, t^+_1)$ and $r_2 = (u_2, v_2, t^-_2, t^+_2)$, $r_1$ precedes $r_2$, noted $r_1 \prec{} r_2$ if $t^+_1 \leq{} t^-_2$ and $u_2 = v_1$.
\end{definition}

Intuitively, the precedence relation among \Rtuples{} tells us that the journeys they represent can be composed, leading to another \Rtuple as follows:

\begin{definition}[Concatenation $\cdot$]\label{def:tuple-composition}
  Given two \Rtuples{} $r_1 = (u_1, v_1, t^-_1, t^+_1)$ and $r_2 = (u_2, v_2, t^-_2, t^+_2)$ such that $r_1 \prec{} r_2$, the concatenation of $r_1$ with $r_2$ is the \Rtuple $r_1 \cdot r_2 = (u_1, v_2, t^-_1, t^+_2)$.
\end{definition}

The natural inclusion among intervals extends to \Rtuples as follows:

\begin{definition}[Inclusion $\subseteq$]\label{def:tuple-inclusion}
  Given two \Rtuples{} $r_1 = (u_1, v_1, t^-_1, t^+_1)$ and $r_2 = (u_2, v_2, t^-_2, t^+_2)$, $r_1 \subseteq{} r_2$ if and only if $u_1 = u_2$, $v_1 = v_2$, and $[t^-_1, t^+_1] \subseteq [t^-_2,t^+_2]$ (\textit{that is,} $t^-_2 \leq{} t^-_1 \leq{} t^+_1 \leq{} t^+_2$).
\end{definition}

If neither $r_1 \subseteq r_2$ nor $r_2 \subseteq r_1$ (or if the vertices are different), then $r_1$ and $r_2$ are called \emph{incomparable}.
Intuitively, if $r_1 \subseteq r_2$, then any of the journeys represented by $r_2$ could be replaced by a (possibly faster) journey represented by $r_1$. More precisely:

\begin{lemma}\label{lemma:discard}
  Let $u$ and $v$ be two nodes in $V$. Let $\mathcal{I}_1 = [t_1^-, t_1^+]$ and $\mathcal{I}_2 = [t_2^-, t_2^+]$ be two subintervals of $\T$ such that $\mathcal{I}_1 \subseteq \mathcal{I}_2$. If $u$ can reach $v$ within $\mathcal{I}_1$, then $u$ can reach $v$ within $\mathcal{I}_2$.
\end{lemma}

\begin{proof}
    The proof is straightforward, we give it for completeness. Let $r$ be the \Rtuple $(u, v, t_1^-, t_1^+)$ and let $\J$ be any of the journeys in $\J(r)$. One can reach $v$ from $u$ within $\mathcal{I}_2$ through the three following steps: (1) wait at $u$ from $t_2^-$ to $t_1^-$, (2) travel from $u$ to $v$ using $\J$, and finally (3) wait at $v$ from $t_1^+$ to $t_2^+$.
\end{proof}

The main consequence of Lemma~\ref{lemma:discard} is that if $r_1 \subseteq r_2$, then $r_2$ is redundant for answering reachability queries from $u$ to $v$.

\begin{definition}[Redundancy]\label{def:tuple-redundancy}
  Let $S$ be a set of \Rtuples{} and let $r \in S$, $r$ is called \emph{redundant} in $S$ if there exists $r' \in S$ such that $r' \subseteq r$. A set with no redundant \Rtuple is called \emph{irredundant}.
\end{definition}

An \Rtuple that is non-redundant in a set is also called \emph{minimal} (in that set).
It is natural to ask what the maximum size of an \emph{irredundant} set of \Rtuples could be, with consequences for the space complexity of a reachability data structure based on \Rtuples. It turns out that this number is always significantly smaller than the number of possible \Rtuples.



\begin{lemma}\label{lemma:number}
  The maximum size of an irredundant set of \Rtuples for \G is $\Theta(n^2\tau)$.
\end{lemma}


\begin{proof}
  First, we prove that the maximum number of pairwise incomparable \Rtuples is $O(n^2\tau)$. Then, we show that this bound is tight, as some graphs induce $\Theta(n^2\tau)$ incomparable \Rtuples.\medskip\\
  \textit{(1) Upper bound:}
  There are $\Theta(n^2)$ ordered pairs of vertices. Thus, it is sufficient to show that for each pair $(u, v)$, the number of incomparable \Rtuples whose starting vertex is $u$ and whose ending vertex is $v$ is $\Theta(\tau)$. Let $S$ be an irredundant set of such \Rtuples, and let $r_1=(u,v,t_1^-,t_1^+)$ and $r_2=(u,v,t_2^-,t_2^+)$ be any two \Rtuples{} in $S$.
  If $t^-_1 = t^-_2$, then either $r_1 \subseteq r_2$ or $r_2 \subseteq r_1$, thus $S$ is redundant (contradiction). As a result, all departure timestamps $t^-_i$ belonging to the \Rtuples{} in $S$ are different, which implies that $|S| \le \tau$.\medskip\\
  \textit{(2) Tightness:}
  Consider the complete temporal graph $\mathcal{K}_{n, \tau}$ on $n$ vertices in which every edge is present in \emph{all} timestamps in $[1, \tau]$. In such a graph, there are consequently $\Theta(n^2\tau)$ contacts, each of which is a trivial journey. Now, observe that either these journeys connect different vertices, or their intervals are incomparable (same duration with different starting times), thus none of them is redundant with the others.
\end{proof}

Given a graph \G and a set $S$ of \Rtuples representing all the journeys of \G, the subset $S' \subseteq S$ of all minimal \Rtuples is called the \emph{representative} \Rtuples of $\G$, denoted by $\R(\G)$. We also write $\R(u,v)$ for those \Rtuples in $\R(\G)$ whose source is $u$ and destination is $v$.
From the proof of Lemma~\ref{lemma:number}, we extract:

\begin{observation}\label{obs:trivial}
  Every contact of \G is present in $\R(\G)$ in the form of a trivial \Rtuple.
\end{observation}

Observation~\ref{obs:trivial} implies that $\R(\G)$ is a \emph{non-lossy} representation, as $\G$ itself is contained in it. The down side is that its space complexity is at least as large as the number of contacts in $\G$. Observe that, up to a constant factor, it can however not be worse than the worst number of contacts, since there may exist up to $\Theta(n^2\tau)$ contacts and irredundant sets cannot exceed this size (Lemma~\ref{lemma:number}). In other words, in dense temporal graphs, the reachability information offered by \Rtuples is essentially free in space.

\subsubsection{Constructive \Rtuples}
\label{sec:constructive}

The data structure considered in this work has four operations, namely \alganame, \algbname, \algcname, and \algdname. The first three operations can be dealt with using only existential \Rtuple. The fourth operation could benefit from storing a small amount of additional information into the \Rtuple.

\begin{definition}[Constructive \Rtuple]
  A constructive \Rtuple{} $r=(u, v, t^-, t^+, w)$ contains the same information as an existential \Rtuple, plus a vertex $w$ such that at least one journey $\J \in \J(r)$ starts with the contact $(u, w, t^-)$. Node $w$ is called the \emph{successor} of $u$ in~$r$ (resp., in $\J$).
\end{definition}

Most of the definitions and lemmas from Section~\ref{ssec:R-tuples} apply unchanged to constructive \Rtuple. In particular, the definition of redundant \Rtuples applies without considering the successor field. Indeed, if two constructive \Rtuples differ only by the successor node, then they are seen as equivalent and any of the two can be discarded. As for the concatenation of two constructive \Rtuples $r_1 = (u_1, v_1, t^-_1, t^+_1, w_1)$ and $r_2 = (u_2, v_2, t^-_2, t^+_2, w_2)$, provided $r_1 \prec r_2$, we additionally require that the resulting \Rtuple adopts the successor of $r_1$ as its own successor; that is, $r_1 \cdot r_2 = (u_1, v_2, t^-_1, t^+_2, w_1)$. For simplicity, whenever constructive \Rtuples are not needed, we describe the algorithms using existential \Rtuples.

\subsection{Timed Transitive Closure}\label{ssec:ttc}


Informally, the timed transitive closure of a temporal graph $\G$ is a multigraph that captures the existence of journeys within all possible time intervals, based on irredundant \Rtuples.

\begin{definition}[Timed transitive closure]\label{def:tc}
  Given a graph $\G$, the \emph{timed transitive closure} of $\G$, noted $TTC(\G)$, is a (static) directed multigraph on the same set of vertices, whose arcs correspond to the representative \Rtuples of $\G$.
\end{definition}

\begin{figure}[h]
  \centering
    \begin{tikzpicture}[scale=3]
      \node[vertex] at (162:8mm) (a) {}
      node[left = -2pt of a] {$a$};
      \node[vertex] at (90:7mm)  (b) {}
      node[above= -2pt of b] {$b$};
      \node[vertex] at (270:2mm) (c) {}
      node[below= -2pt of c] {$c$};
      \node[vertex] at (378:8mm) (d) {}
      node[right= -2pt of d] {$d$};
      \draw (a) edge[arc] node[above left] {$[2, 3]$} (b);
      \draw (b) edge[arc, bend left=15] node[above left = 3pt and -10pt] {$[4, 5]$} (d);
      \draw (b) edge[arc, bend right=15] node[above right = -5pt and -5pt] {$[1, 2]$} (d);
      \draw (a) edge[arc, bend left=15] node[below left= -3pt and -9pt] {$[4, 5]$} (c);
      \draw (c) edge[arc, bend left=15] node[below right= 5pt and -7pt] {$[4, 5]$} (a);
      \draw (c) edge[arc] node[below right] {$[5, 6]$} (d);

      \draw (a) edge[arc, bend left=6] node[above] {$[2, 5]$} (d);
      \draw (a) edge[arc, bend right=9] node[below] {$[4, 6]$} (d);
    \end{tikzpicture}
    ~~~~~~~~
    \begin{tikzpicture}[scale=3]
      \node[vertex] at (162:8mm) (a) {}
      node[left = -2pt of a] {$a$};
      \node[vertex] at (90:7mm)  (b) {}
      node[above= -2pt of b] {$b$};
      \node[vertex] at (270:2mm) (c) {}
      node[below= -2pt of c] {$c$};
      \node[vertex] at (378:8mm) (d) {}
      node[right= -2pt of d] {$d$};
      \draw (a) edge[arc] node[above left] {$([2, 3], b)$} (b);
      \draw (b) edge[arc, bend left=20] node[above left = 3pt and -15pt] {$([4, 5], d)$} (d);
      \draw (b) edge[arc, bend right=20] node[above right = -4pt and -3pt] {$([1, 2], d)$} (d);
      \draw (a) edge[arc, bend left=20] node[below left= -3pt and -11pt] {$([4, 5], c)$} (c);
      \draw (c) edge[arc, bend left=20] node[below right= 5pt and -11pt] {$([4, 5], a)$} (a);
      \draw (c) edge[arc] node[below right] {$([5, 6], d)$} (d);

      \draw (a) edge[arc, bend left=6] node[above] {$([2, 5], b)$} (d);
      \draw (a) edge[arc, bend right=9] node[below] {$([4, 6], c)$} (d);
    \end{tikzpicture}
  \caption{Timed transitive closure $\TTC(\G)$ of the temporal graph $\mathcal{G}$ in Figure~\ref{fig:temporal-graph}, considering $\delta = 1$. On the left, the version with existential \Rtuple{}, whose intervals are depicted by labels. On the right, the version with constructive \Rtuples, depicting also the successor. 
  }
    \label{fig:timed-transitive-closure}
\end{figure}
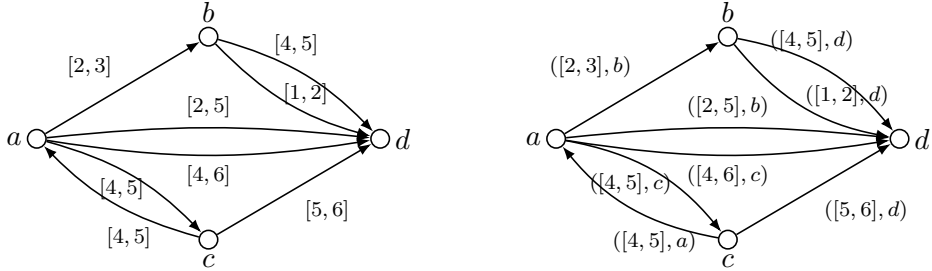

Figure~\ref{fig:timed-transitive-closure} shows two examples of \TTCs, both corresponding to the temporal graph of Figure~\ref{fig:temporal-graph} (one for existential \Rtuples{}, the other for constructive \Rtuples{}). Algorithmically, a \TTC provides most of the support needed to realize the high-level operations of our data structure. For example, the operation \algb{} amounts to testing if there exists an arc whose associated \Rtuple is $(u, v, t^-, t^+)$ with $[t^-, t^+] \subseteq [t_1, t_2]$. The operation \algc{} can be realized by performing such a test for every pair of vertices.
The operation \alga{} reduces to adding a new arc to the \TTC{}(\G) if no smaller interval already captures this information. If the new arc is added, then some other arcs may become redundant and should be removed, some others may also be created by composition. This operation is therefore the most critical. Finally, if constructive \Rtuples are used, then an actual journey may be reconstructed quite efficiently from \TTC{}(\G) when \algd{} is called, by retrieving a constructive \Rtuple{} $(u, v, t^-, t^+, w)$ such that $[t^-, t^+] \subseteq [t_1, t_2]$ and unfolding the corresponding journey inductively, by replacing $u$ with the successor vertex $w$ and $t^-$ with $t^- + \delta$ in each step.

All algorithms for these operations are described in Section~\ref{sec:data-structure}. Before doing so, we present an explicit encoding of \TTCs based on adjacency matrices and binary search trees (BST). In order for the high-level algorithms to remain independent from this particular choice, we define a set of primitives for manipulating the \TTC, that are used by the high-level algorithms of Section~\ref{sec:data-structure}.

\subsubsection{Encoding the TTC}
\label{subsec:data-layout}

We encode the \TTC by an $n \times n$ matrix, in which every entry $(i,j)$ points to a self-balanced binary search tree (BST) denoted by $T(i,j)$. The nodes in this tree contain all the time intervals corresponding to \Rtuples in $\R(i, j)$.
From Lemma~\ref{lemma:number}, we know that a tree $T(u, v)$ contains up to $\tau$ nodes.
In addition, all these intervals are incomparable, thus one can use any of their boundaries (departure or arrival) as the sorting key of the BST.
Note that retrieving $T(u, v)$ within the matrix takes constant time, as the cells of a matrix are directly accessed.
Also recall that finding the largest key below (resp. the smallest key above) a certain value takes $O(\log{\tau})$ time. Similarly, inserting a new element (in our case, an interval) takes $O(\log\tau)$ time. Finally, observe that several types of BST (e.g. red-black trees) can self-balance without impacting the asymptotic cost of insertions.

We provide the following low-level operations for manipulating \TTCs: (1)
\Call{find\mund{}next}{$u, v, t$} returns the containing the earliest interval $[t^-, t^+]$ in $T(u, v)$ such that $t^- \geq t$, if any, and nil otherwise; symmetrically, (2)
\Call{find\mund{}previous}{$u, v, t$} returns the node containing the latest interval $[t^-, t^+]$ in $T(u, v)$ such that $t^+ \leq t$, if any, and nil otherwise;
finally, (3) \Call{insert}{$u, v, t^-, t^+$} inserts a new node containing the interval $[t^-, t^+]$ in $T(u, v)$ and performs some operations for maintaining the property that all intervals in $T(u,v)$ are minimal.

Let us now describe the algorithms that perform these operations, along with their time complexities.
The algorithm for \Call{find\mund{}next}{$u, v, t$} searches $T(u, v)$ recursively, by comparing $t$ with the departure $t^-$ of the current node interval $[t^-, t^+]$.
If $t^-$ is equals to or greater than $t$, then the current node is a candidate answer.
The algorithm then compares the current node candidate and the previous one, and keeps the one containing the smallest (earliest) $t^-$, then it descends the left child.
Otherwise, if $t^-$ is smaller than $t$, it simply descends the right child.
As soon as a leaf is reached (and visited), the algorithm returns the current candidate as the answer.
The algorithm for \Call{find\mund{}previous}{$u, v, t$} works symmetrically.
The time complexities of both algorithms correspond to the depth of the tree, which is $O(\log{\tau})$.

\begin{figure}
  \newcommand{\drawhighlighted}[3]{
    \draw{(#1, #3) edge[semithick] (#2, #3)};
    \draw{(#1, #3 -0.25) edge[dashed, very thin] (#1, 0.25)};
    \draw{(#2, #3 -0.25) edge[dashed, very thin] (#2, 0.25)};
    \node[] at (#1 +2.5, #3 +0.25) {\footnotesize$\mathcal{I}$};
  }
  \newcommand{\drawnotincludes}{
    \draw{(1, 0) edge[semithick] (5, 0)};
    \draw{(2, -0.25) edge[semithick] (6, -0.25)};
    \draw{(6, -1.25) edge[semithick] (10, -1.25)};
    \draw{(7, -1.50) edge[semithick] (11, -1.50)};
    \draw{(8, -1.75) edge[semithick] (12, -1.75)};
    \draw{(9, -2.00) edge[semithick] (13, -2.00)};
  }
  \centering
  \begin{subfigure}[b]{0.3\columnwidth}
    \centering
    \begin{tikzpicture}[scale=0.4]
      \drawnotincludes{}
      \drawhighlighted{5}{7}{-3.50};
      \draw (3, -0.50) edge[semithick, red] (7, -0.50);
      \node[text=red] at (7.6, -0.25) {\footnotesize$\mathcal{I}_1$};
      \draw (4, -0.75) edge[semithick, red] (8, -0.75);
      \draw (5, -1.00) edge[semithick, red] (9, -1.00);
      \node[text=red] at (9.6, -0.75) {\footnotesize$\mathcal{I}_2$};
    \end{tikzpicture}
    \caption{Find}\label{fig:find-insertion}
  \end{subfigure}~%
  \begin{subfigure}[b]{0.3\columnwidth}
    \centering
    \begin{tikzpicture}[scale=0.4]
      \drawnotincludes{}
      \drawhighlighted{5}{7}{-3.50};
    \end{tikzpicture}
    \caption{Remove}\label{fig:remove-insertion}
  \end{subfigure}~%
  \begin{subfigure}[b]{0.3\columnwidth}
    \centering
    \begin{tikzpicture}[scale=0.4]
      \drawnotincludes{}
      \phantom{\drawhighlighted{5}{7}{-3.50};}
      \draw (5, -0.625) edge[semithick, red] (7, -0.625);
      \node[text = red] at (7.5, -0.375) {\footnotesize$\mathcal{I}$};
    \end{tikzpicture}
    \caption{Insert}\label{fig:insert-insertion}
  \end{subfigure}
  \caption{Basic steps to perform $\Call{insert}{u, v, \mathcal{I}}$ where $\mathcal{I} = [t_1, t_2]$. First, in~(a), an algorithm must find the candidate intervals that could become redundant after the inserting $\mathcal{I}$. These intervals are exactly the ones between $\mathcal{I}_1 = \Call{find\mund{}previous}{u, v, t^+}$ and $\mathcal{I}_2 = \Call{find\mund{}next}{u, v, t^-}$. Note that there are cases in which $\mathcal{I}_1$ or $\mathcal{I}_2$ do not exist.
  Next, in~(b), all intervals $\mathcal{I}'$ between (and including) $\mathcal{I}_1$ and $\mathcal{I}_2$ such that $\mathcal{I}' \subseteq \mathcal{I}$ must be removed.
  Finally, in~(c), the algorithm inserts $\mathcal{I}$ in the correct place.}\label{fig:insertion}
\end{figure}
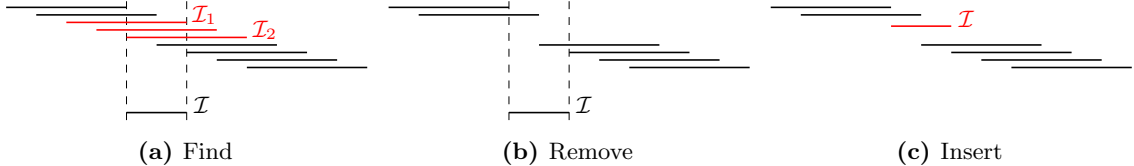

The algorithm for \Call{insert}{$u, v, t^-, t^+$} finds and removes any potential node with interval $\mathcal{I}_i$ such that $[t^-, t^+] \subseteq \mathcal{I}_i$, then it inserts a new node containing $[t^-, t^+]$ using a standard BST insertion.
Figure~\ref{fig:insertion} gives a linear representation of the intervals in $T(u, v)$ while performing this operation.
A naive implementation of this operation would consist of searching and removing each corresponding node independently. However, this would lead to a complexity of $O(d\log\tau)$ time, where $d$ is the number of nodes removed, that is up to $O(\tau)$. We use a non-standard approach that makes it feasible in $O(d + \log\tau)$ time only. The strategy is to identify in $T(u, v)$ the nodes containing, respectively, the boundary intervals $\mathcal{I}_1 = \Call{find\mund{}previous}{u, v, t^+}$ and $\mathcal{I}_2 = \Call{find\mund{}next}{u, v, t^-}$, which correspond to the first and last nodes to be removed (note that the parameters are indeed $t^+$ and $t^-$, not the reverse).
Then, every node containing intervals in this range is removed using the technique outlined in the proof of Lemma~\ref{lem:cost-insert}.

\begin{lemma}
  \label{lem:cost-insert}
  In the worst case, the cost of the \textsc{insert} operation is $O(d + \log{\tau})$, where $d$ is the number of elements removed from $T(u, v)$. The amortized cost of an insertion is $O(\log{\tau})$.
\end{lemma}

\begin{proof}
  The range of intervals to be removed is characterized by two boundary intervals $\I_1$ and $\I_2$, which can be found by calling both \textsc{find\mund{}next} and \textsc{find\mund{}previous} a single time, which takes $O(\log\tau)$ time. The final insertion of the input interval in the BST also takes $O(\log\tau)$. The difficult part is thus the removal of redundant intervals prior to this insertion (illustrated in Figure~\ref{fig:insertion}). Let $d$ be the number of intervals in the deletion range.
  We start by recalling the main ideas of range deletion in a BST (see~\cite{range-deletion} for a pedagogical explanation), then we discuss their use in the particular case of balanced BSTs, and finally we explain why the claimed cost is correct despite the fact that balance may be lost after the deletion.
  Let $\I_A$ be the common ancestor of $\I_1$ and of $\I_2$ (possibly equal to one of these).
  The process can be split into two phases, the first one is to walk upward from $\I_1$ to $\I_A$ and the second is to walk upward (again) from $\I_2$ to $\I_A$. As either walk proceeds, potential deletions are performed in intermediate tree nodes. Some of these deletions remove the node itself, replacing it by one of its child in constant time. The intermediate branches are cut without being explored. The final cost of $O(d + \log\tau)$ actually follows from the cumulated length of both walks.
  The down side with this technique (which may explain why it is not standard in balanced BSTs) is that the resulting tree may have lost its balance if most of the nodes are deleted. However, as the depth cannot increase, and as we only need that it remains of $O(\log\tau)$ for subsequent use of the tree, the solution is good enough for our needs. Finally, observe that the number of intervals removed cannot exceed the number of intervals previously inserted, which is why the amortized cost of an insertion remains of $O(\log{\tau})$.
\end{proof}

Additionally, we define the following basic operations:
\begin{itemize}
\item $\mathcal{N}^*_{out}(u)$ : Returns the set of vertices $\{ v_1, v_2, \ldots, v_k \}$ such that there exists at least one arc from $u$ to $v_i$ in the \TTC
\item $\mathcal{N}^*_{in}(u)$ : Returns the set of vertices $\{ v_1, v_2, \ldots, v_l \}$ such that there exists at least one arc from $v_i$ to $u$ in the \TTC
\end{itemize}
Both operations can be realized in $O(n)$ time, through traversing the corresponding row (resp. column) of the matrix and testing if the corresponding tree is empty.

\section{The four operations}\label{sec:data-structure}

In this section, we describe the algorithms which perform the four operations of our data structure, previously described in Section~\ref{sec:primitives}. These operations are \algb{}, \algc{}, \alga{}, and (optionally) \algd{}.
For simplicity, the first three algorithms are presented using existential \Rtuples only (however they are straightforwardly adaptable to constructive \Rtuples). All the algorithms rely on the primitives defined in Section~\ref{subsec:data-layout} for manipulating the \TTC in an abstract way.

\subsection{Reachability and connectivity queries}\label{ssec:query-operation1}


The algorithm for performing \algb{} is straightforward. It consists of testing whether $T(u, v)$ contains at least one interval that is included in $[t_1, t_2]$.
This can be done by retrieving $[t^-, t^+] = \Call{find\mund{}next}{u, v, t_1}$  and checking that $t^+ \leq t_2$.
Therefore, the cost of this algorithm reduces essentially to that of the operation \Call{find\mund{}next}{$u, v, t_1$}, which takes $O(\log\tau)$ time.
We note that, if $[t_1, t_2] = \mathcal{T}$, then it is sufficient to verify (in constant time) that $T(u, v)$ is not empty.
Regarding the operation \algc{}, a simple way of answering it is to call \algb{} for every pair of vertices, with a resulting time complexity of $O(n^2 \log{\tau})$. It seems plausible that this strategy is not optimal and could be improved in the future.

\newpage
\subsection{Update operation}\label{ssec:update-operation}

The algorithm for \alga{} manages the insertion of a new contact $(u, v, t)$ in the data structure, where $(u, v) \in{E}$ and $t \in{\tau}$. To start, the interval corresponding to the trivial journey from $u$ to $v$ over $[t,t+\delta]$ is inserted in $T(u,v)$ using the \textsc{insert} primitive. (Recall that this primitive encapsulate the removal of redundant intervals in $T(u,v)$, if any.)
Then, the core of the algorithm consists of computing the indirect consequences of this insertion for the other vertices. Namely, if a vertex $w^-$ could reach $u$ before time $t$ with latest departure $t^-$ and $v$ could reach another vertex $w^+$ after time $t+\delta$ with earliest arrival $t^+$, it follows that $w^-$ can now reach $w^+$ over interval $[t^-, t^+]$. Our algorithm consists of enumerating these compositions and inserting them in the \TTC.
Interestingly, for each predecessor $w^-$ of $u$, only the \emph{latest} interval ending before $t$ in $T(w^-, u)$ needs to be considered. The reason
is that in order to compose an earlier journey $\J$ with the new contact, we need to wait at $u$ until time $t$. Thus, even though some other journeys started earlier, it would have to wait at $u$ and it would thus eventually arrive at the same time (based on a non-minimal interval). Based on this property, our algorithm only searches for the latest interval preceding $t$ for each predecessor of $u$ and the earliest interval exceeding $t+\delta$ for each successor of $v$.

The details are given in Algorithm~\ref{alg:1}, whose behavior is as follows.
In Line 1, the algorithm inserts the interval $[t, t + \delta]$ into $T(u, v)$, which corresponds to the trivial journey induced by the new contact.
In Lines 2 to 7, for every vertex $w^- \in \mathcal{N}^*_{in}(u)$, it finds the latest interval $[t^-, \_]$ in $T(w^-, u)$ that arrives before time $t$ (inclusive) and inserts the composition $[t^-, t + \delta]$ into $T(w^-, v)$.
For the same reasons as above, the algorithm only needs considering inserting $[t^-, t + \delta]$ because every other possible composition would contain it as a subinterval.
In Lines 8 to 11, for every vertex $w^+ \in \mathcal{N}^*_{out}(v)$, the algorithm finds the earliest interval $[\,\_\,, t^+] $ in $T(v, w^+)$ that leaves $v$ after time $t + \delta$ (inclusive), and inserts the composition $[t, t^+]$ into $T(u, w)$.
In the same way, every other possible composition would contain $[t, t^+]$ as a subinterval.
Finally, in Lines 12 to 14, for all $w^- \in \mathcal{N}^*_{in}(u)$ and $w^+ \in \mathcal{N}^*_{out}(v)$, it inserts the composition $[t^-, t^+]$ into $T(w^-, w^+)$.
In order to optimize this last step, the algorithm only considers the subset of $\mathcal{N}^*_{in}$ whose reachability to $v$ has been impacted by the new contact, thanks to a dedicated storage $D$ computed in Line 7.


\begin{algorithm*}
  \caption{\alga}\label{alg:1}
  \begin{algorithmic}[1]
    \Require{$t \in{\mathcal{T}}, u,v \in V$ with $u \neq{} v$}
    \State{$\Call{insert}{u, v, t, t + \delta}$}
    \State{$D \gets{} \{\}$}
    \ForAll{$w^- \in{\mathcal{N}^*_{in}(u)}$}
      \State{$[t^-, \_] \gets{} \Call{find\mund{}previous}{w^-, u, t}$}
      \If{$t^- \ne nil$}
        \State{$\Call{insert}{w^-, v, t^-, t + \delta}$}
        \State{$D \gets{} D \cup{} (w^-, t^-)$}
      \EndIf{}
    \EndFor{}
    \ForAll{$w^+ \in{\mathcal{N}^*_{out}(v)}$}
      \State{$[\_\,, t^+] \gets{} \Call{find\mund{}next}{v, w^+, t + \delta}$}
      \If{$t^+ \ne nil$}
        \State{$\Call{insert}{u, w^+, t, t^+}$}
        \ForAll{$(w^-, t^-) \in{D}$}
          \If{$w^- \neq w^+$}
            \State{$\Call{insert}{w^-, w^+, t^-, t^+}$}
          \EndIf{}
        \EndFor{}
      \EndIf{}
    \EndFor{}
  \end{algorithmic}
\end{algorithm*}

\begin{theorem}
  The update operation has amortized time complexity $O(n^2 \log\tau)$. In the worst case, a single update operation costs $O(n^2 \tau)$ time.
\end{theorem}

\newcommand{\costInsert}{\ensuremath{d + \log\tau}\xspace}

\begin{proof}
  An \textsc{insert} operation is performed in Line~1.
  The loop from Line 3 to 7 iterates over $O(n)$ vertices and makes at most one insertion for each. The loop from Line 8 to 14 iterates over $O(n)$ vertices, and for each one, iterates in a nested way over $O(n)$ vertices. For each resulting pair, it performs at most one \textsc{insert} operation. The latter clearly dominates the overall cost of the algorithm, with a cost of $O(n^2)$ times the cost of the \textsc{insert} operation, the latter being of amortized time $O(\log\tau)$ and otherwise of $O(d + \log\tau)$ time (Lemma~\ref{lem:cost-insert}), with $d=O(\tau)$ in the worst case.
\end{proof}

\subsection{Journey reconstruction}\label{ssec:query-operation2}

The algorithm for the operation \algd{} reconstructs a journey from vertex $u$ to vertex $v$ whose contact timestamps must be contained in $[t_1, t_2]$.
As explained in Section~\ref{sec:constructive}, existential \Rtuples can be augmented by a \emph{successor} field that indicates which vertex comes next in (at least one of) the journeys represented by the \Rtuple. This information is very useful for reconstruction and has a negligible cost (asymptotically speaking). Concretely, one can make the nodes of the BST store the successor field in addition to the interval. The low-level operations for manipulating the TTC (see Section~\ref{subsec:data-layout}) are unaffected, neither are the query and update algorithms in a significant way.
The only subtlety is that when two intervals (nodes) are composed, the successor field of the resulting node corresponds to the successor field of the first node (this was already discussed in terms of \Rtuples in Section~\ref{sec:timed-transitive-closure}).

The goal of the algorithm is thus to reconstruct a journey by unfolding the intervals and successor fields. Details are given in Algorithm~\ref{alg:2}.
\begin{algorithm*}
  \caption{\algd}\label{alg:2}
  \begin{algorithmic}[1]
    \Require{$[t_1, t_2] \subseteq{} \mathcal{T}, u,v \in V, u \neq{} v$}
    \State{$([t^-, t^+], w) \gets{} \Call{find\mund{}next}{u, v, t_1}$}
    \Comment{node augmented with successor}
    \If{the return value is $nil$ or $t^+ \leq t_2$}
      \State{\Return{$nil$}}
      \Comment{no interval contained in $[t_1, t_2]$ in $T(u,v)$}
    \EndIf{}
    \State{$\mathcal{J} \gets{} \{(u, w, t^-)\}$}

    \While{$w \neq v$}
      \State{$([t, \_\,], w') \gets{} \Call{find\mund{}next}{w, v, t^- + \delta}$}
      \State{$\mathcal{J} \gets{} \mathcal{J} \cdot \{(w, w', t)\}$}
      \State{$w \gets{} w'$}
      \State{$t^- \gets{} t$}
    \EndWhile{}

    \State{\Return{$\mathcal{J}$}}
  \end{algorithmic}
\end{algorithm*}
The first step (Lines 1 to 3) is to retrieve a node in $T(u,v)$ whose interval is contained within $[t_1, t_2]$ if one exists. If several choices exist, the earliest is selected (through calling the \textsc{find\mund{}next} primitive). Then, the algorithm iteratively replaces $u$ with the successor and searches for the next interval until the successor is $v$ itself (Lines 5 to 9), adding gradually the corresponding contacts to a journey $\J$ (Line 4 and Line 7), which is ultimately returned in Line 10.

\begin{theorem}
  Algorithm~\ref{alg:2} has time complexity $O(k\log\tau)$, where $k$ is the length of the resulting journey.
\end{theorem}

\begin{proof}
  The algorithm calls \textsc{find\_next} in Line~1. After that, it is known whether a journey can be reconstructed. If so, a journey prefix $\J$ is initialized with the first contact of the reconstructed journey (indeed, such a contact must exist due to the minimality of the interval). Then, in the loop from Line~5 to Line~9, the algorithm extends $\J$ by one contact for each call to \textsc{find\_next} until $\J$ contains the entire journey. Overall, \textsc{find\_next} is thus called as many times as the length of the reconstructed journey, which corresponds to $O(|\J| \log\tau)$ time. The costs of the other operations are clearly dominated by this cost.
\end{proof}

\subsubsection{Properties of the reconstructed journeys}

In general, several journeys may exist that satisfy the query parameters. We observe that the specific choices made in Algorithm~\ref{alg:2} imply additional properties.

\begin{lemma}
  \label{lem:foremost}
  \textup{The journey $\J$ which is returned by Algorithm~\ref{alg:2} is a \emph{foremost} journey in the requested interval (i.e., it arrives at the earliest possible time at $v$). Furthermore, among all the possible foremost journeys, it is also a \emph{fastest} journey (i.e., the difference between departure time and arrival time is minimized).}
\end{lemma}
\begin{proof}
  The fact that $\J$ is a foremost journey follows from the call to \textsc{find\_next} in Line~1. Indeed, the interval returned by this call corresponds to the earliest departure from $u$, which happens to also correspond to the earliest arrival at $v$ because the stored intervals are incomparable. $\J$ thus achieves the earliest possible arrival time at $v$ in the given interval. And since all the stored intervals are \emph{minimal} (i.e. they do not contain smaller reachability intervals), it also follows that $departure(\J)$ is as late as possible among all the journeys arriving in $v$ at time $arrival(\J)$, which means $\J$ is as fast as possible among all foremost journeys.
\end{proof}

Let us insist that Lemma~\ref{lem:foremost} does not imply that $\J$ is both foremost and fastest in the requested interval. It only states that $\J$ is a foremost journey, and a fastest one \emph{among} the possible foremost journeys. Even faster journeys might exist in the requested interval, arriving later at $v$. The above property is however already convenient, e.g. in communication networks, where a message would arrive at destination as early as possible, while (secondarily) traveling for as little time as possible.

\section{Conclusion and open questions}\label{sec:conclusions}

We presented in this paper an incremental data structure to solve the dynamic connectivity problem in temporal graphs.
Our data structure places a high priority on the query time, by answering reachability questions in time $O(\log\tau)$.
Based on the ability to retrieve reachability information for particular time intervals, it supports the insertion of contacts in a non-chronological order in $O(n^2\log\tau)$ amortized time (deterministic worst-case $O(n^2\tau)$ time) and makes it possible to reconstruct efficiently foremost journeys within a given time interval, i.e., in time $O(k\log\tau)$, where $k$ is the size of the resulting journey. Our algorithms exploit the special features of non-redundant (minimal) reachability information, which we represent through the concept of \Rtuples. The core of our data structure, namely the \emph{timed transitive closure} (\TTC), is itself essentially a collection of irredundant \Rtuples, whose size (and that of the data structure itself) cannot exceed $O(n^2\tau)$.

The theory of \Rtuples, initiated in this paper, poses a number of further questions, some of which are of independent interest, some leading to possible improvements of the presented algorithms. For example, do \Rtuples{} involving different pairs of vertices possess further interdependence which may reduce the space needed to maintain \TTCs{}? More generally, how restricted are \TTCs intrinsically? On the practical side, can we improve the insertion time for new contacts by using another low-level structure than a balanced BST? Could the notion of contacts be generalized to contacts of arbitrary duration?
Finally, designing efficient data structures for the decremental and the fully-dynamic versions of this problem, with \emph{unsorted} contact insertion and deletion, seems to represent both a significant challenge and a natural extension of the present work, one that would certainly develop further our common understanding of temporal reachability.


\begin{paragraph}{Acknowledgements}
This study was financed in part by Fundação de Amparo à Pesquisa do Estado de Minas Gerais (FAPEMIG) and the Coordena\c{c}\~{a}o de Aperfei\c{c}oamento de Pessoal de N\'{i}vel Superior - Brasil (CAPES) - Finance Code 001* - under the ``CAPES PrInt program'' awarded to the Computer Science Post-graduate Program of the Federal University of Uberl\^{a}ndia, as well as the Agence Nationale de la Recherche through ANR project ESTATE (ANR-16-CE25-0009-03).
\end{paragraph}

\printbibliography{}

@article{xuan2003computing,
  author = {Xuan, B. Bui and Ferreira, A. and Jarry, A.},
  title = {Computing shortest, fastest, and foremost journeys in dynamic networks},
  journal = {International Journal of Foundations of Computer Science},
  volume = {14},
  number = {02},
  pages = {267-285},
  year = {2003},
}

@article{barjon2014testing,
  author = {Matthieu Barjon and Arnaud Casteigts and Serge Chaumette and Colette Johnen and Yessin M. Neggaz},
  title = {Testing temporal connectivity in sparse dynamic graphs},
  journal = {CoRR},
  volume = {abs/1404.7634},
  year = {2014},
  archiveprefix = {arXiv},
}

@inproceedings{wu2016reachability,
  author = {H. {Wu} and Y. {Huang} and J. {Cheng} and J. {Li} and Y. {Ke}},
  booktitle = {2016 IEEE 32nd International Conference on Data Engineering (ICDE)},
  title = {Reachability and time-based path queries in temporal graphs},
  year = {2016},
  volume = {},
  number = {},
  pages = {145-156},
}

@inproceedings{agrawal1989efficient,
  author = {Agrawal, R. and Borgida, A. and Jagadish, H. V.},
  title = {Efficient management of transitive relationships in large data and knowledge bases},
  year = {1989},
  isbn = {0897913175},
  publisher = {Association for Computing Machinery},
  address = {New York, NY, USA},
  booktitle = {Proceedings of the 1989 ACM SIGMOD International Conference on Management of Data},
  pages = {253–262},
  numpages = {10},
  location = {Portland, Oregon, USA},
  series = {SIGMOD ’89},
}

@inproceedings{haixun2006dual,
  author = { {Haixun Wang} and {Hao He} and {Jun Yang} and P. S. {Yu} and J. X. {Yu}},
  booktitle = {22nd International Conference on Data Engineering (ICDE'06)},
  title = {Dual labeling: answering graph reachability queries in constant time},
  year = {2006},
  volume = {},
  number = {},
  pages = {75-75},
}

@article{cohen2003reachability,
  author = {Cohen, Edith and Halperin, Eran and Kaplan, Haim and Zwick, Uri},
  title = {Reachability and distance queries via 2-hop labels},
  journal = {SIAM Journal on Computing},
  volume = {32},
  number = {5},
  pages = {1338-1355},
  year = {2003},
  doi = {10.1137/S0097539702403098},
}

@inproceedings{zhu2014reachability,
  author = {Zhu, Andy Diwen and Lin, Wenqing and Wang, Sibo and Xiao, Xiaokui},
  title = {Reachability queries on large dynamic graphs: a total order approach},
  year = {2014},
  isbn = {9781450323765},
  publisher = {Association for Computing Machinery},
  address = {New York, NY, USA},
  booktitle = {Proceedings of the 2014 ACM SIGMOD International Conference on Management of Data},
  pages = {1323–1334},
  numpages = {12},
  location = {Snowbird, Utah, USA},
  series = {SIGMOD ’14},
}

@inproceedings{seufert2013ferrari,
  author = {S. {Seufert} and A. {Anand} and S. {Bedathur} and G. {Weikum}},
  booktitle = {2013 IEEE 29th International Conference on Data Engineering (ICDE)},
  title = {FERRARI: flexible and efficient reachability range assignment for graph indexing},
  year = {2013},
  volume = {},
  number = {},
  pages = {1009-1020},
}

@article{wei2018reachability,
  title = {Reachability querying: an independent permutation labeling approach},
  author = {Wei, Hao and Yu, Jeffrey Xu and Lu, Can and Jin, Ruoming},
  journal = {The VLDB Journal},
  volume = {27},
  number = {1},
  pages = {1--26},
  year = {2018},
  publisher = {Springer},
}

@article{tang2010characterising,
  author = {Tang, John and Musolesi, Mirco and Mascolo, Cecilia and Latora, Vito},
  title = {Characterising temporal distance and reachability in mobile and online social networks},
  year = {2010},
  issue_date = {January 2010},
  publisher = {Association for Computing Machinery},
  address = {New York, NY, USA},
  volume = {40},
  number = {1},
  issn = {0146-4833},
  journal = {SIGCOMM Comput. Commun. Rev.},
  month = jan,
  pages = {118–124},
  numpages = {7},
}

@inproceedings{whitbeck2012temporal,
  author = {Whitbeck, John and Dias de Amorim, Marcelo and Conan, Vania and Guillaume, Jean-Loup},
  title = {Temporal reachability graphs},
  year = {2012},
  isbn = {9781450311595},
  publisher = {Association for Computing Machinery},
  address = {New York, NY, USA},
  booktitle = {Proceedings of the 18th Annual International Conference on Mobile Computing and Networking},
  pages = {377–388},
  numpages = {12},
  location = {Istanbul, Turkey},
  series = {Mobicom ’12},
}

@article{williams2016spatio,
  author = {Williams, Matthew J. and Musolesi, Mirco },
  title = {Spatio-temporal networks: reachability, centrality and robustness},
  journal = {Royal Society Open Science},
  volume = {3},
  number = {6},
  pages = {160196},
  year = {2016},
}

@inbook{cacciari1996atemporal,
  author = "Cacciari, Leo and Rafiq, Omar",
  editor = "Dembi{\'{n}}ski, Piotr and {\'{S}}redniawa, Marek",
  title = "A temporal reachability analysis",
  booktitle = "Protocol Specification, Testing and Verification XV: Proceedings of the Fifteenth IFIP WG6.1
    International Symposium on Protocol Specification, Testing and Verification, Warsaw, Poland, June 1995",
  year = "1996",
  publisher = "Springer US",
  address = "Boston, MA",
  pages = "35--49",
  isbn = "978-0-387-34892-6",
}

@inproceedings{wu2017mining,
  author = {G. {Wu} and Y. {Ding} and Y. {Li} and J. {Bao} and Y. {Zheng} and J. {Luo}},
  booktitle = {2017 IEEE 33rd International Conference on Data Engineering (ICDE)},
  title = {Mining spatio-temporal reachable regions over massive trajectory data},
  year = {2017},
  volume = {},
  number = {},
  pages = {1283-1294},
}

@article{martensen2017spatio,
  author = {Martensen, Alexandre Camargo and Saura, Santiago and Fortin, Marie-Josee},
  title = {Spatio-temporal connectivity: assessing the amount of reachable habitat in dynamic landscapes},
  journal = {Methods in Ecology and Evolution},
  volume = {8},
  number = {10},
  pages = {1253-1264},
  year = {2017},
}

@article{bryce2007atutorial,
  title = {A tutorial on planning graph based reachability heuristics},
  volume = {28},
  number = {1},
  journal = {AI Magazine},
  author = {Bryce, Daniel and Kambhampati, Subbarao},
  year = {2007},
  month = mar,
  pages = {47},
}

@inproceedings{bedogni2018temporal,
  author = {L. {Bedogni} and M. {Fiore} and C. {Glacet}},
  booktitle = {IEEE INFOCOM 2018 - IEEE Conference on Computer Communications},
  title = {Temporal reachability in vehicular networks},
  year = {2018},
  volume = {},
  number = {},
  pages = {81-89},
}

@misc{vera2016querying,
  title = {Querying evolving graphs with portal},
  author = {Vera Zaychik Moffitt and Julia Stoyanovich},
  year = {2016},
  archiveprefix = {arXiv},
  primaryclass = {cs.DB},
}

@article{casteigts2012time,
  title={Time-varying graphs and dynamic networks},
  author={Casteigts, Arnaud and Flocchini, Paola and Quattrociocchi, Walter and Santoro, Nicola},
  journal={International Journal of Parallel, Emergent and Distributed Systems},
  volume={27},
  number={5},
  pages={387--408},
  year={2012},
  publisher={Taylor \& Francis}
}

@inproceedings{betsy2007spatio,
  author = "George, Betsy and Kim, Sangho and Shekhar, Shashi",
  editor = "Papadias, Dimitris and Zhang, Donghui and Kollios, George",
  title = "Spatio-temporal network databases and routing algorithms: a summary of results",
  booktitle = "Advances in Spatial and Temporal Databases",
  year = "2007",
  publisher = "Springer Berlin Heidelberg",
  address = "Berlin, Heidelberg",
  pages = "460--477",
  isbn = "978-3-540-73540-3",
}

@article{zeng2014visualizing,
  author = {W. {Zeng} and C. {Fu} and S. M. {Arisona} and A. {Erath} and H. {Qu}},
  journal = {IEEE Transactions on Visualization and Computer Graphics},
  title = {Visualizing mobility of public transportation system},
  year = {2014},
  volume = {20},
  number = {12},
  pages = {1833-1842},
}

@inproceedings{hasan2011making,
  title = {Making sense of time: timeline visualization for public transport schedule},
  author = {Hasan, Khandaker Tabin and Noori, Sheak Rashed Haider and SalamИ, Abdus and Kabir, Md Anwarul},
  booktitle = {Symposium on Human-Computer Interaction and Information Retrieval (HCIR 2011), Washington},
  year = {2011},
}

@article{hurter2014bundled,
  author = {C. {Hurter} and O. {Ersoy} and S. I. {Fabrikant} and T. R. {Klein} and A. C. {Telea}},
  journal = {IEEE Transactions on Visualization and Computer Graphics},
  title = {Bundled visualization of dynamicgraph and trail data},
  year = {2014},
  volume = {20},
  number = {8},
  pages = {1141-1157},
}

@misc{enright2018deleting,
  title = {Deleting edges to restrict the size of an epidemic in temporal networks},
  author = {Jessica Enright and Kitty Meeks and George B. Mertzios and Viktor Zamaraev},
  year = {2018},
  archiveprefix = {arXiv},
  primaryclass = {cs.DS},
}

@inproceedings{rozenshtein2016reconstructing,
  author = {Rozenshtein, Polina and Gionis, Aristides and Prakash, B. Aditya and Vreeken, Jilles},
  title = {Reconstructing an epidemic over time},
  year = {2016},
  isbn = {9781450342322},
  publisher = {Association for Computing Machinery},
  address = {New York, NY, USA},
  booktitle = {Proceedings of the 22nd ACM SIGKDD International Conference on Knowledge Discovery and Data Mining},
  pages = {1835–1844},
  numpages = {10},
  location = {San Francisco, California, USA},
  series = {KDD ’16},
}

@inbook{xiao2018reconstructing,
  author = {Han Xiao and Polina Rozenshtein and Nikolaj Tatti and Aristides Gionis},
  title = {Reconstructing a cascade from temporal observations},
  booktitle = {Proceedings of the 2018 SIAM International Conference on Data Mining},
  chapter = {},
  pages = {666-674},
}

@INPROCEEDINGS{xiao2018robust,
  author={H. {Xiao} and C. {Aslay} and A. {Gionis}},
  booktitle={2018 IEEE International Conference on Data Mining (ICDM)},
  title={Robust cascade reconstruction by steiner tree sampling},
  year={2018},
  volume={},
  number={},
  pages={637-646},
}

@MISC{range-deletion,
    TITLE = {Remove range of keys from Binary Search Tree in O(s+h)},
    AUTHOR = {Stefano Leucci},
    HOWPUBLISHED = {Computer Science Stack Exchange},
    NOTE = {URL:https://cs.stackexchange.com/q/123535 (version: 2020-04-02)},
    EPRINT = {https://cs.stackexchange.com/q/123535},
    URL = {https://cs.stackexchange.com/q/123535}
}

@article{LVM18,
  title={Stream graphs and link streams for the modeling of interactions over time},
  author={Latapy, Matthieu and Viard, Tiphaine and Magnien, Cl{\'e}mence},
  journal={Social Network Analysis and Mining},
  volume={8},
  number={1},
  pages={1--29},
  year={2018},
  publisher={Springer}
}

\end{document}